\documentclass{article}
\usepackage{graphicx}
\usepackage{amsmath}
\usepackage{amssymb}
\usepackage{color}
\usepackage{a4wide}

\newcommand{\R}{\ensuremath{{\mathbb R}}}
\newcommand{\Z}{\ensuremath{{\mathbb Z}}}
\newcommand{\N}{\ensuremath{{\mathbb N}}}
\newcommand{\eps}{\ensuremath{\varepsilon}}

\newcommand{\X}{\mathcal{X}}
\newcommand{\nxt}{\ensuremath{\mathrm{next}}}

\newenvironment{proof}%
{\noindent\emph{Proof}.\hspace{1ex}}%
{\hfill\unitlength=0.18ex%
  \begin{picture}(12,12)
    \put(1,1){\framebox(9,9){}}
    \put(1,4){\framebox(6,6){}}
  \end{picture}\linebreak
}

\newtheorem{theorem}{Theorem}
\newtheorem{lemma}{Lemma}

\newcommand{\COMMENT}[1]{}

\def\coolfigure#1#2#3{%
  \begin{figure}[tbh]
    \topsep 0pt
    \begin{center}
      \leavevmode#1
    \end{center}
    \caption{#2}
    \label{#3}
  \end{figure}
}
\newcommand{\epsfigure}[4][3cm]{\coolfigure%
  {\includegraphics[clip,width=#1]{#2}}{#3}{#4}}

\begin{document}

\title{Computing $k$-Centers On a Line}

\author{
  Peter Brass\thanks{Department of Computer Science, City College, New  York, USA. Email: {\tt peter@cs.ccny.cuny.edu}} \and
  Christian Knauer\thanks{Institute of Computer Science, Free University Berlin,  Germany. Email: {\tt knauer@inf.fu-berlin.de}}\and
  Hyeon-Suk Na\thanks{School of Computing, Soongsil University, Seoul, Korea.  Email: {\tt hsnaa@ssu.ac.kr}. Corresponding author. Supported by Korean Research Foundation Grant(KRF-2007-531-D00018), and the
Soongsil University Research Fund.} \and
Chan-Su Shin\thanks{School of Electronics and Information Engineering,  Hankuk~University~of~Foreign~Studies, Korea. Email: {\tt cssin@hufs.ac.kr}. Supported by Korean Research Foundation Grant(KRF-2007-311-D00764), and the
HUFS Research Fund.}  \and
  Antoine Vigneron\thanks{INRA, UR 341 Math\'ematiques
et Informatique Appliqu\'ees,
78352 Jouy-en-Josas, France.
Email: \texttt{antoine.vigneron@jouy.inra.fr}.}
}

\date{}
\maketitle

\begin{abstract}
  In this paper we consider several instances of the \emph{$k$-center on a   line problem} where the goal is, given a set of points $S$  in the plane and a parameter $k\geq 1$, to find  $k$ disks with centers on a line $\ell$ such that their union covers $S$ and  the maximum radius of the disks is minimized. This problem is a constraint version of the well-known $k$-center problem in which the centers are constrained to lie in a particular region such as a segment, a line, and a polygon.

  We first consider the simplest version of the problem where the line  $\ell$ is given in advance; we can solve this problem in $O(n\log^2 n)$  time. We then investigate the cases  where only the orientation of the line $\ell$ is fixed   and where the line $\ell$ can be arbitrary.  We can solve these problems in $O(n^2 \log^2 n)$ time and  in $O(n^4\log^2 n)$ expected time, respectively.  For the last two problems, we present $(1+\eps)$-approximation  algorithms, which run in $O(\frac{1}{\eps}n\log^2 n)$ time and $O({1\over \eps^2}n\log^2 n)$ time, respectively.
\end{abstract}

\section{Introduction}

A common type of facility location or clustering problem is the $k$-center problem, which is defined as follows: Given a set $S$ of $n$ points in a metric space and a positive integer $k$, find a set of $k$ supply points such that the maximum distance between a point in $S$ and its nearest supply point is minimized. For the cases of the $L_2$ and $L_\infty$-metric the problem is usually referred to as the Euclidean and rectilinear $k$-center problem respectively. The cluster centers (i.e., the supply points) can be seen as an approximation of the set $P$. Drezner~\cite{D95} describes many variations of the facility location problem and their numerous applications.

$k$-center problems (as well as many other clustering problems) arise naturally in many applications (like, e.g., in the shape simplification or in the construction of bounding volume hierarchies).
They can be formulated as geometric optimization problems and, as such, they have been studied extensively in the field of computational geometry.
Efficient polynomial-time algorithms have been found for the planar $k$-center problem when $k$ is a small constant~\cite{Chan99,Epp97}.  Also, the rectilinear $2$-center problem can be solved in polynomial time, even when $d$ is part of the input~\cite{Meg90}.  However both the Euclidean and rectilinear (decision) problems are NP-hard, for $d=2$ when $k$ is part of the input~\cite{FPT81,MS84}, while the Euclidean 2-center and rectilinear 3-center are NP-hard when $d$ is part of the input~\cite{Meg90}. Hwang et al.~\cite{hwang} gave an $n^{O(\sqrt{k})}$-time algorithm for the Euclidean $k$-center problem in the plane. Agarwal and Procopiuc~\cite{agarwal02} presented an $n^{O(k^{1-1/d})}$-time algorithm for solving the $k$-center problem in $\mathbb{R}^d$ under $L_2$ and $L_\infty$-metrics, and a simple $(1+\eps)$-approximation algorithm with running time $O(n\log k)+(k/\eps)^{O(k^{1-1/d})}$.

\paragraph{Problems.}
In this paper, we consider a variation of the $k$-center problem
with potential applications in wireless network design: Given $n$
sensors which are modeled as points, we want to locate $k$ base stations (or servers), which are modeled as centers of disks, for receiving the signal
from the sensors. The sensors operate by self-contained batteries,
but the servers should be connected to a power line, so they have
to lie on a straight line which models the power line.
Thus we define the \emph{$k$-center on a line problem} as follows:
\emph{Given a set $S$ of $n$ points (i.e., sensors) and an integer $k\geq 1$, find $k$ disks with centers (i.e., base stations) on a line such that the union of the disks covers $S$ and the maximum radius of the disks is minimized.} We will also study several variants of the problem depending on whether the orientation of the center line is fixed or not.

\paragraph{Related results.}
The problem of finding the $k$-centers on a line is a constrained version of the standard $k$-center problems where the centers are constrained to lie in a specific region such as a line, a segment, and a polygonn.

For $k = 1$, Megiddo~\cite{Meg86} showed that the 1-center constrained to be on a given line can be computed in $O(n)$ time; our problem is a direct extension of this for $k > 1$. Hurtado et al.~\cite{hurtado} considered the 1-center in a convex region instead of the line, and showed that the 1-center lying in the convex polygon of $m$ vertices can be computed in $O(n+m)$ time. If the 1-center is restricted to lie in a set of simple polygons with a total of $m$ edges, then Bose et al.~\cite{bose08} proved that it can be computed in $O((n+m)\log n)$ time, which is an improvement to the previous quadratic time algorithm~\cite{bose}. They also showed how to preprocess the $n$ points in $O(n\log n)$ time and $O(n)$ space such that for any query segment the 1-center lying on the segment can be reported in $O(\log n)$ time.

For $k \geq 2$, there are a few results~\cite{das,roy,chansu} which have been done mostly for the base station placement problem in wireless sensor network. Das et al.~\cite{roy} studied more constrained problems where the $k$ centers lie on a specific edge of a convex polygon and the corresponding disks cover all $n$ vertices of the polygon. Their algorithm runs in $O(\min(n^2, nk\log n))$ time. In~\cite{roy, chansu}, they considered different constrained 2-center problems in which the centers lie in a convex polygon~\cite{chansu} or on a pair of specified edges of the convex polygon~\cite{roy}, but the two disks should cover all points in the polygon, i.e., the convex polygon itself.

Alt et al.~\cite{alt} studied a similar problem where the goal is to
minimize the sum of the radii (instead of minimizing the maximum
of the radii); in fact, they minimized the cost function of the form $\sum_i r_i^\alpha$ for any fixed $\alpha \geq 1$ under any fixed $L_p$ metric, where $r_i$ represents the radius of each disk. In their model $k$ is not part of the input. They presented an algorithm that runs in $O(n^c\log n)$ time to compute the optimal solution for a fixed line where $c = 2$ for $\alpha = 1$ and $3\leq c\leq 4$ for $\alpha > 1$. For a horizontal line moving freely, they proved that the value of the optimal solution (as well as the location of the optimal horizontal line) cannot be expressed by radicals, and they gave an $(1+\eps)$-approximation algorithm requiring $O({\frac{1}{\eps}}n^3\log n)$ time. They also presented an $(1+\eps)$-approximation algorithm for arbitrary line which takes $O({\frac{1}{\eps^2} }n^5\log n)$ time.

\paragraph{Our results.}
We investigate the $k$-center problem for arbitrary $k\geq 1$ where the $k$ centers are collinear and should cover a set of $n$ points. In Section~\ref{sec:fixed} we first consider the simplest version of the problem where the line $\ell$ is fixed (we will assume w.o.l.g. that the centers of the $k$ disks must lie on the $x$-axis). We can solve this problem in $O(n\log^2 n)$ time. Next, we consider two problems: In Section~\ref{sec:orienfix}, we look at the case where only the orientation of the line $\ell$ is fixed (we will assume w.l.o.g. that $\ell$ is horizontal), and solve it in $O(n^2 \log^2 n)$ time.  In Section~\ref{sec:arbit}, we study the most general case where the line $\ell$ can have an arbitrary orientation. We solve this problem in $O(n^4\log^2 n)$ expected time. For these two problems, we also present $(1+\eps)$-approximation algorithms, which run in $O(\frac{1}{\eps}n\log^2 n)$ time and $O({1\over \eps^2}n\log^2 n)$ time, respectively. These results are the first on the $k$-center problem on a line for $k\geq 2$. Note here that $k$ is part of the input but not involved in the running time, and all algorithms also work for the problems under any fixed $L_p$ metric. The results are summarized in Table~\ref{tab:results}.

\begin{table}
\label{tab:results}
\begin{minipage}{\textwidth}
\begin{tabular}{p{7cm}||c|c}\hline\hline
 $k$-center problems on a line & Exact algorithm & $(1+\eps)$-approximation \\ \hline
 Centers on a fixed line & $O(n\log^2n)$
 \footnote{$O(n\log n)$ time for $L_1$ and $L_\infty$ metric.} & - \\ \hline
 Centers on a line with fixed orientation & $O(n^2 \log^2 n)$
 & $O((1/ \eps)\,n \log^2n)$ \\ \hline
 Centers on a line with arbitrary orientation & $O(n^4 \log^2 n)$
 & $O((1/ \eps^2)\,n \log^2 n)$ \\ \hline\hline
\end{tabular}
\end{minipage}
 \caption{Summary of the results.}
\end{table}

\paragraph{Preliminaries.}

Let $S = \{s_1, \ldots, s_n\}$ be the set of input points, where
$s_i := (x_i, y_i)$. We assume that no two points in $S$ have the
same $x$-coordinate, no four points lie on the same circle and
that they are sorted in non-decreasing $x$-order. Let $\ell$ be
the line on which the disk centers have to lie. Denote by $D(p,r)$
the disk centered at $p$ with radius $r$; if the radius $r$ is
fixed, we denote this disk by $D(p)$.

\section{$k$-centers on a fixed line}\label{sec:fixed}

We first consider the simplest version of the problem where the
line $\ell$ is fixed.  Without loss of generality, we assume that
$\ell$ is the $x$-axis. Let $y_\mathrm{max} := \max_i |y_i|$. Let
$r^\star$ be the minimum radius such that there exist $k$ disks of
that radius with centers on the $x$-axis and with union covering
$S$. Then it is clear that $r^\star \geq y_\mathrm{max}$. To find
the minimum radius $r^\star$, we will perform a binary search,
combining the results of~\cite{cole,matousek} with an algorithm
for the following decision problem:

 \emph{Given $r \geq y_\mathrm{max}$, decide if there exist $k$ disks of radius $r$ with centers on the $x$-axis and with union covering $S$.}

\subsection{The decision algorithm}

Let $X_i(r)$ be the intersection of $D(s_i, r) \cap \ell$, i.e., $X_i(r) = [a_i(r), b_i(r)]$ is an interval on the $x$-axis with the two endpoints $a_i(r) \leq b_i(r)$. A \emph{piercing set} of the interval set $\X(r) = \{ X_i(r) \}$ is a set of points on $\ell$ such that every interval in $\X(r)$ contains at least one point in the piercing set. We call the minimum cardinality of a piercing set for $\X(r)$ the \emph{minimum piercing number} $k(r)$ of $\X(r)$. This number can be computed in $O(n\log n)$ time by selecting a piercing point at the leftmost right endpoint, removing all the intervals containing the piercing point, and repeating the same process for the remaining intervals until all the intervals are removed, c.f.~\cite{katz,chan}.

It is clear that $S$ can be covered by the union of $k$ disks of radius $r$ with centers on $\ell$ if and only if the minimum piercing number $k(r)$ is not larger than $k$. Thus we can answer the decision problem in $O(n\log n)$ time; if the right endpoints (and the left endpoints) of $X_i(r)$ are sorted then we can even do it in $O(n)$ time.

\subsection{Computing the optimal radius}\label{subsec:fixed}

If the piercing number at $r = y_\mathrm{max}$ is no more than
$k$, we find that $r^\star = y_\mathrm{max}$. So, in what follows,
we assume that $r^\star > y_\mathrm{max}$, and that the intervals
$X_i(r)= [a_i(r), b_i(r)]$ are sorted in non-decreasing order of
their right endpoints $b_i(r)$. As $r$ increases, $a_i(r)$
decreases and $b_i(r)$ increases. We can interpret $a_i(r')$ and
$b_i(r')$ as the $x$-coordinates of intersection points of the
hyperbola $\tilde{X}_i: r^2 - (x-x_i)^2 = y_i^2$ with the
horizontal line $r=r'$ in the $(x,r)$-plane, as illustrated in
Figure~\ref{fig:fixedline}. Since no bisector between $s_i$ and
$s_j$ is horizontal, every bisector intersects the $x$-axis $\ell$
once and thus two hyperbolas $\tilde{X}_i$ and $\tilde{X}_j$ meet
exactly once. If we consider $\tilde{X}_i$ as two curve segments
by cutting it at its lowest point, then we have $2n$ curves with
the property that any two of them intersects at most once.  We
will denote the set of these curves (which are also called
pseudoline segments) by $\tilde{\X}$.

\epsfigure[\textwidth]{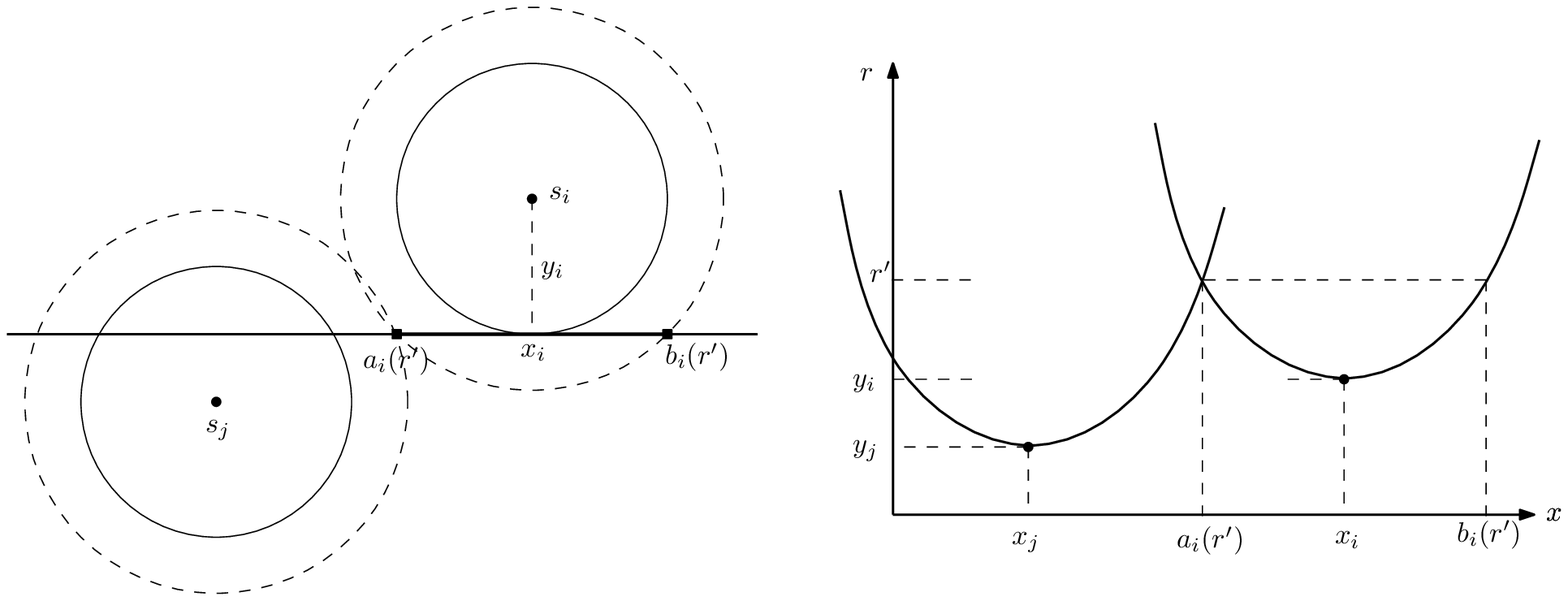}{Intervals defined on a fixed
line.}{fig:fixedline}

As $r$ increases, the minimum piercing number $k(r)$ for the
interval system $\X(r)$ is non-increasing. Our goal is to find the
radius $r^\star$ such that $k(r^\star)$ is the largest integer
with $k(r^\star) \leq k$. Since $k(r)$ can change (i.e., decrease)
only when the relative order of the endpoints of the intervals in
$\X(r)$ changes, we know that $r^\star$ is a radius where two
hyperbolas from $\tilde{\X}$ meet. A straightforward way to find
the radius $r^\star$ is to compute all the radii at which two
hyperbolas meet, sort them, and do a binary search over these
radii with the decision algorithm of the previous  subsection;
sorting all the $O(n^2)$ candidate radii takes $O(n^2\log n)$ time
and the binary search requires $O(n\log^2 n)$ time because it
invokes the decision algorithm $O(\log n)$ times. Thus the total
running time is $O(n^2\log n)$.

We can improve the running time to $O(n\log^2 n)$ by using the
result of~\cite{cole,matousek} which was developed for the
slope-selection problem.
In the slope-selection problem we are given a set of $n$ points in
the plane and an integer $m$, and we have to determine the line of
the $m$-th smallest slope among all the lines spanned by two input
points. In the dual setting, this problem is equivalent to finding
the intersection point with the $m$-th smallest $y$-coordinate
among all intersection points of the lines (dual to the input
points). This problem can be solved in $O(n\log n)$ deterministic
time~\cite{cole,matousek}. Since the curves in $\tilde{\X}$ we are
dealing with are pseudoline segments, we can adopt these
algorithms immediately. Thus we can do one step of the binary
search in $O(n\log n)$ time: choose the median radius $r_{\rm
med}$ by the algorithm of~\cite{cole,matousek}, and decide if
$r^\star \leq r_{\rm med}$ by the decision algorithm of the
previous section. As a result, we can find $r^\star$ in $O(n\log^2
n)$ total time. The result is summarized in the following:
\begin{theorem}\label{th:0dof}
  Let $S$ be a set of $n$ points in $\R^2$ and $k \in \N$.  Let $r^\star$
  be the minimum radius such that there exist $k$ disks of that radius
  with centers on the $x$-axis and with union covering $S$.
  Then $r^\star$ and such disks can be computed in $O(n\log^2 n)$ deterministic
time.
\end{theorem}

\paragraph{Remark.}

We can compute $r^\star$ in $O(n\log n)$ time if we use the $L_1$
or the $L_\infty$ metric.  In these cases the functions $a_i(r)$
and $b_i(r)$ are linear functions with a slope of $-45^\circ$ and
$+45^\circ$, respectively. There are no intersections among the
functions $a_i(r)$ and no intersections among the functions
$b_i(r)$. Consequently the sorted sequence of the $a_i(r)$'s
(resp., $b_i(r)$'s) defined at $r = y_\mathrm{max}$ remains the
same for any $r \geq y_\mathrm{max}$. So once we have the sorted
sequences at the beginning of the algorithm, we do not need to
sort the endpoints again when we make decision for a fixed radius
during the binary search; this results in a linear-time decision
algorithm for a fixed radius $r$.

A radius $r$ will be potentially tested during the binary search only if
$b_i(r) = a_j(r)$ for some $i \neq j$. We can represent these radii
implicitly in a doubly sorted matrix as follows:
Let $a_{\pi(1)}, a_{\pi(2)}, \ldots, a_{\pi(n)}$ be the increasing sequence
of $a_i(r)$ for some permutation $\pi$ at $r= y_\mathrm{max}$, and let $b_{\sigma(1)}, b_{\sigma(2)}, \ldots, b_{\sigma(n)}$ be the decreasing sequence of $b_i(r)$ for some permutation $\sigma$ at $r = y_\mathrm{max}$.
We define $r(\pi(i),\sigma(j))$ as the radius $r$ such that
$a_{\pi(i)}(r) = b_{\sigma(j)}(r)$; if $r(\pi(i),\sigma(j)) <
y_\mathrm{max}$ then we set $r(\pi(i),\sigma(j)):= y_\mathrm{max}$.  It is easy to see that $r(\pi(i),\sigma(j)) < r(\pi(i),\sigma(l))$ for any $l > j$, and $r(\pi(i),\sigma(j)) > r(\pi(h),\sigma(j))$ for any $h < i$.
Consequently, if we put the $r(\pi(i),\sigma(j))$ into an $n\times n$
matrix with rows (ordered by $\pi$) representing $a_{\pi(i)}$ and columns
(ordered by $\sigma$) representing $b_{\sigma(j)}$ we get a matrix in which
each row and each column is totally ordered. With this doubly sorted matrix,
we can search the $m$-th smallest entry in the matrix in $O(n)$ time~\cite{fredjohn} (in particular we need to access only $O(n)$ matrix entries). Using the $O(n)$-time decision algorithm outlined above, we can perform the binary search for $r^\star$ in $O(n\log n)$ total time.

\begin{theorem}\label{th:0dofl1}
  Let $S$ be a set of $n$ points in $\R^2$ and $k \in \N$.  Let $r^\star$
be the minimum radius such that there exist $k$ $L_1$-disks (resp.
$L_\infty$-disks) of that radius, with centers on the $x$-axis and
with union covering $S$.  Then $r^\star$ and such disks can be
computed in $O(n\log n)$ deterministic time.
\end{theorem}

\section{$k$-centers on a line with fixed orientation\label{sec:orienfix}}

We now consider the case where only the orientation of the line
$\ell$ is fixed.  Without loss of generality, we may assume that
$\ell$ is horizontal. Let $r^\star$ be the minimum radius such
that there exist $k$ disks of that radius with centers on {\em a
horizontal line} and with union covering $S$. For $y_0 \in \R$,
denote the horizontal line $y=y_0$  by $\ell(y_0)$.

Like in the previous case, we first develop an algorithm for the
following decision problem, and then find the minimum radius
$r^\star$ by using this algorithm and the techniques
in~\cite{cole,matousek}:

\emph{Given $r>0$, decide if there exist $k$ disks of radius $r$
  with centers on a horizontal line $\ell(y_r)$ and with union
  covering $S$.}

\subsection{The decision algorithm}\label{subsec:orienfix}

Consider the disks $D(s_i):=D(s_i,r)$ of radius $r$ around the
points $s_i \in S$.  As the line $\ell(y)$ moves from $y = -\infty$ to $y = +\infty$, we maintain the minimum piercing number
$k(y)$ for the intervals $X_i(y) := D(s_i)\cap \ell(y)$. Since all
the disks must intersect $\ell(y)$, it is actually sufficient to
sweep $\ell(y)$ from the topmost bottom point of the disks to the
bottommost top point of the disks. Denote the $y$-coordinates of
these two points by $y_s$ and $y_t$, respectively. Then for any
$y$ from $y_s $ to $y_t$, the interval $X_i(y)$ of any point $s_i$
is non-empty. Set $\X(y):=\{ X_i(y) \mid s_i\in S\}$.

As the line $\ell(y)$ moves from  $y=y_s$ to $y=y_t$, the minimum
piercing number $k(y)$ changes only when the relative order of the
endpoints of two intervals changes. Thus the events correspond to
intersections of circles bounding the $D(s_i)$'s, and the total
number of the events is $O(n^2)$. If we have a dynamic data
structure to maintain the minimum piercing number of intervals for
each event in $O(T)$ (amortized) time, we can handle all the
events in $O(n^2T)$ time.


Chan and Mahmood~\cite{chan} describe a data structure to maintain
the minimum piercing number of a set of $n$ intervals under
insertions with $O(\log n)$ amortized time per insertion. Using
this data structure, we can maintain $k(y)$ during the sweep with
$O(\log n)$ amortized time per event, so we can handle all the
events in $O(n^2\log n)$ time.

\paragraph{Maintaining the piercing number.}
For completeness, we briefly explain how the data structure
proposed by Chan and Mahmood~\cite{chan} can be applied to our
problem. For a fixed $y$, we consider the intervals $X_i = [a_i,
b_i]$ for $s_i\in S$ (omitting `$y$' in the notation), together
with two dummy intervals $X_0 = [-\infty, -\infty]$ and $X_{n+1} =
[\infty, \infty]$. The greedy algorithm to compute the minimum
piercing number (as described in the previous section) chooses the
right endpoints of the intervals as piercing points. If $b_i$ and
$b_j$ for $i < j$ are two consecutive piercing points chosen by
the greedy algorithm, then $b_j = \min_l\{b_l\mid a_l > b_i\}$.
Now we define $\nxt(X_i) := X_j$ if $b_j = \min_l \{ b_l \mid a_l
> b_i\}$, and partition the intervals into groups each of which is
a maximal set of intervals with a common ``next'' value.
Chan and Mahmood observe that the intervals in a group appear in
consecutive order when all $(n+2)$ intervals are sorted according
to their right endpoints. Thus we can define a weighted tree $T$
of intervals with the ``next'' relation such that a pair of
consecutive vertices in a group is connected by an edge of weight
$0$, and if $X_i$ is the last interval in its group, then the
vertex for $X_i$ is connected to the vertex for $\nxt(X_i)$ by an
edge of weight $1$. As a consequence, the minimum piercing number
corresponds to the total weight of the path from $X_0$ to
$X_{n+1}$ in $T$.

The tree $T$ is implemented with a data structure for dynamic
trees~\cite{dtree} that supports link, cut and path-length queries
in $O(\log n)$ amortized time. The intervals in each group $G$ are
maintained in a balanced search tree $B_G$ ordered by the right
endpoints so that the operations such as search, split, and
concatenation can be done in logarithmic time. All the intervals
are also stored in a priority search tree $Q$ ordered by the left
endpoints with priorities defined by the right endpoints. This
allows us to find the next interval for a given interval in
$O(\log n)$ time. These priority search trees support insertion
and deletion of intervals in $O(\log n)$ time each. For further
details, refer to~\cite{chan}.

As the sweep line $\ell(y)$ moves from $y_s$ to $y_t$, we maintain the
intervals on $\ell(y)$ in the dynamic tree $T$.
First, at $y = y_s$, we construct the data structures of all $n$ intervals
by simply inserting them in $O(n \log n)$ time~\cite{chan}. The line $\ell(y)$ stops at each event $y_s < y \leq y_t$, which is the $y$-coordinate of the
intersection of the circles bounding two disks $D(s_i)$ and $D(s_j)$ for
some $i$ and $j$; we here assume that $s_i$ is in the left of $s_j$.
When these two circles intersect, the relative order of the endpoints of $X_i(y)=[a_i(y),b_i(y)]$ and $X_j(y)=[a_j(y),b_j(y)]$
changes just after $y$. Then we have four different cases that the order change. Let $y^+ := y + \eps$ and $y^- := y - \eps$ for a small $\eps > 0$ (see Figure~\ref{fig:sweepevents}).

\epsfigure[12cm]{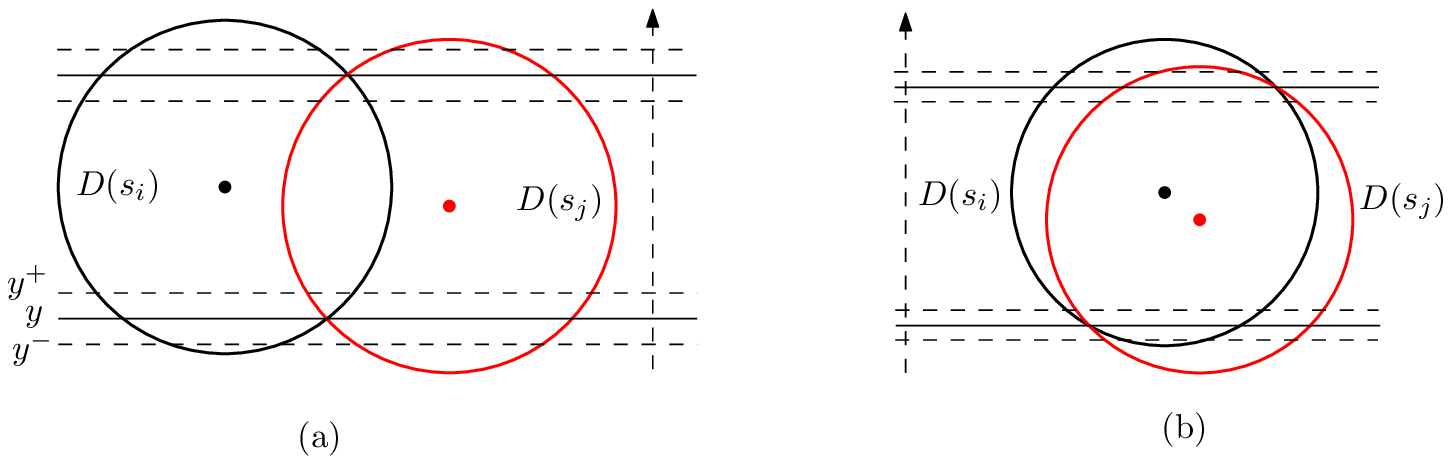}{The four possible events.
(a) Case (i) at the lower intersection point and Case (ii) at the upper intersection point.
(b) Case (iii) at the left intersection point and Case (iv) at the right intersection point.}{fig:sweepevents}

\paragraph{Case (i) $b_i(y^-) < a_j(y^-)$ and $b_i(y^+) > a_j(y^+)$.}
Two intervals $X_i$ and $X_j$ start to overlap after $y$. If the
next interval of $X_i$ at $y^-$ is not $X_j$, then we have nothing
to do. Otherwise, $X_j$ is no longer the next interval of $X_i$ at
$y^+$. By the definition of the group, $X_i$ must be the last
interval in the group $G$ that $X_i$ belongs to. First we find the
new next interval $X_l$ of $X_i$ at $y^+$ using $Q$, where $l \neq
j$. Next we split the group $G$ into two groups $G_1 := G\setminus
\{X_i\}$ and $G_2 := \{X_i\}$. Then, at $y^+$, the next interval
for $G_1$ is still $X_j$, and the next one for $G_2$ is now $X_l$.
So we merge $G_2$ into the group $G'$ that $X_l$ belongs to at
$y^-$. For these changes, we update $T$, $B_G$, and $B_{G'}$.

\paragraph{Case (ii) $b_i(y^-) > a_j(y^-)$ and $b_i(y^+) < a_j(y^+)$.}
Two intervals $X_i$ and $X_j$ start to be disjoint after $y$. The
interval $X_j$ is not the next interval for $X_i$ at $y^-$ because
they overlap, but $X_j$ can be the next interval for $X_i$ at
$y^+$. If it is indeed, then we know from the definition of the
group that $X_i$ must be the first interval of the group $G$ of
$X_i$ at $y^-$. We split $G$ into $G_1 :=\{ X_i \}$ and $G_2:=
G\setminus \{ X_i \}$. The next interval of $G_1$ is now $X_j$,
thus $G_1$ is merged with the group $G'$ whose next interval is
$X_j$ at $y^-$. As did in Case (i), we update $T$, $B_G$, and
$B_{G'}$ to reflect the changes.

\paragraph{Case (iii) $a_i(y^-) > a_j(y^-)$ and $a_i(y^+) < a_j(y^+)$.}
In this case the relative order of the left endpoints is reversed
after $y$, i.e., $a_i(y^+) < a_j(y^+)$. For this change, we have nothing to do.
\paragraph{Case (iv) $b_i(y^-) < b_j(y^-)$ and $b_i(y^+) > b_j(y^+)$.}
The relative order of the right endpoints is reversed after $y$, i.e., $b_i(y^+) > b_j(y^+)$. For the group $G$ whose next interval is $X_i$ at $y^-$, the next interval must be changed from $X_i$ to $X_j$. So we merge $G$ into the group $G'$ of $X_j$. For this change, we update $T$, $B_G$, $B_{G'}$ and $Q$ properly.\\

As a result, we can update all the data structures in $O(\log n)$
amortized time per event, so we can solve the decision problem for
a fixed radius in $O(n^2\log n)$ time.

\subsection{Computing the optimal line and radius}

We now describe how to find the optimal line $\ell(y^\star)$ and
radius $r^\star$ by a binary search that uses the decision
algorithm for a fixed radius. To run a binary search, we need a
discrete candidate set for the optimal radius $r^\star$, and we
first study necessary conditions for the optimal line
$\ell(y^\star)$ and radius $r^\star$.

The optimal line $\ell(y^\star)$ must be immobilized, in the sense
that if we translate it either upward or downward then
the radius of the disks should be increased in order to cover $S$.
So in every optimal configuration, there must be at least two
points of $S$ on the circles that immobilize $\ell(y^\star)$,
as illustrated in Figure~\ref{fig:optimal}. In the first two configurations
$\ell(y^\star)$ is fixed by two or three points on one bounding
circle, and in the other configurations it is fixed by at most
four points on two bounding circles.

%
%

\epsfigure[\textwidth]{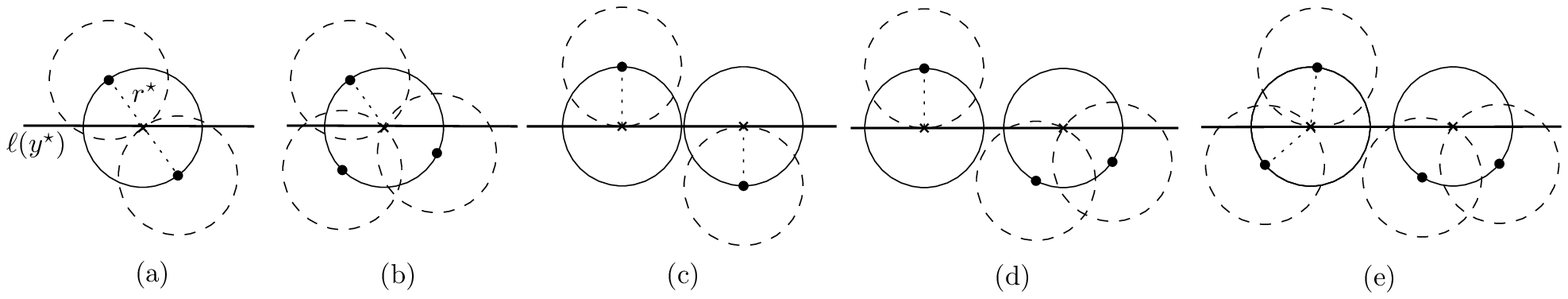}{Five configurations
of the optimal line and radius; symmetric ones are omitted.}{fig:optimal}

%

From these optimal configurations, we obtain $O(n^2)$ candidates
for the optimal radius as follows: consider the discs $D_{ij}$
with $s_i$ and $s_j$ on its bounding circle. Their centers move
along the bisector of $s_i$ and $s_j$, and thus in the
$(y,r)$-plane we can define the function $r = r_{ij}(y)$ of the
radius of $D_{ij}$ with center of $y$-coordinate $y$. It is easy
to see that $r_{ij}$ is a unimodal function, i.e., $r_{ij}$ is
decreasing for $y \le (y_i+y_j)/2$ and increasing for $y \ge
(y_i+y_j)/2$. If we define $r_{ii}(y)$ as the distance of $s_i$ to
the line $\ell(y)$, we get another unimodal function $r_{ii}(y)$
(which is decreasing for $y \le y_i$ and increasing for $y \ge
y_i$).

We now split $r_{ij}$ into two monotone pieces: a decreasing piece
$r^-_{ij}$ and an increasing piece $r^+_{ij}$, and set
$\mathcal{R}^- := \{r^-_{ij}\mid 1 \le i \le j \le n\}$,
$\mathcal{R}^+ := \{r^+_{ij}\mid 1 \le i \le j \le n\}$, and
$\mathcal{R} := \mathcal{R}^-\cup\mathcal{R}^+$. Then every
optimal configuration illustrated in Figure~\ref{fig:optimal}
corresponds to an intersection point of two pieces, one from
$\mathcal{R}^-$ and the other from $\mathcal{R}^+$; the first case
(a) corresponds to the point ${r^-_{ij}} \cap {r^+_{ij}}$ for a
pair of $s_i$ and $s_j$, the second case (b) corresponds to the
point ${r^-_{ij}} \cap {r^+_{jh}}$ for a triple of $s_i$, $s_j$,
and $s_h$, and the remaining cases (c), (d), and (e) correspond to
the points ${r^-_{ii}} \cap {r^+_{jj}}$, ${r^-_{ii}} \cap
{r^+_{jh}}$, and  ${r^-_{ij}} \cap {r^+_{hl}}$ for points $s_i$,
$s_j$, $s_h$, and $s_l$, respectively. Therefore, the optimal
radius $r^\star$ is the $r$-coordinate of one of the intersection
points between $\mathcal{R}^-$ and $\mathcal{R}^+$. Both
$\mathcal{R}^-$ and $\mathcal{R}^+$ consist of $O(n^2)$ curves, so
the number of intersection points is $O(n^4)$.

Consequently, we can determine $r^\star$ by performing a binary
search (discriminated by the decision algorithm) on the radii
associated with the vertices of the arrangement of the curves in
$\mathcal{R}$. But the complexity of the arrangement is $O(n^4)$,
so we find a way to compute the median of these radii without
computing the arrangement explicitly. Since any two functions in
$\mathcal{R}$ intersect at most once, we can pick the median of
the radii from the arrangement in $O(n^2\log n)$ time, by using
the modified slope selection algorithm~\cite{cole,matousek} as we
did in Section~\ref{subsec:fixed}. With this median of the radii,
we run the decision algorithm given in
Section~\ref{subsec:orienfix} in $O(n^2\log n)$ time. Thus we can
perform a step in the binary search in $O(n^2\log n)$ time, and
$r^\star$ can be found in $O(n^2\log^2 n)$ time.

We just proved the following:
\begin{theorem}\label{th:orienfix}
  Let $S$ be a set of $n$ points in $\R^2$ and $k \in \N$.  Let $r^\star$
  be the minimum radius such that there exist $k$ disks of that radius
  with centers on a horizontal line and with union covering $S$.  Then
  $r^\star$ and such disks can be computed in $O(n^2\log^2 n)$ deterministic
  time.
\end{theorem}

\paragraph{Remarks.} This algorithm can be immediately applied to the
problem in the $L_1$-metric. For the $L_\infty$-metric, it is trivially
solved in $O(n\log n)$ time because the optimal horizontal line is the
middle line between the lowest and highest points.


\section{$k$-centers on a line with any orientation}\label{sec:arbit}

We now consider the case where the line $\ell$ can have arbitrary
orientation and position. For $k = 2$, the problem is equivalent
to the standard two-center problem~\cite{Chan99,Epp97}, so we
assume that $k > 2$. Let $r^\star$ be the minimum radius such that
there exist $k$ disks of that radius with centers on a line and
with union covering $S$. We will find $r^\star$ in a similar way
as before: we first design a deterministic algorithm for the
following decision problem and then perform a randomized binary
search over some candidate set of radii:

\emph{Given $r>0$, decide if there exist $k$ disks of radius $r$
  with centers on a line and with union covering $S$.}

The decision algorithm for fixed $r$ runs in $O(n^4\log n)$ time,
and the randomized algorithm for finding the optimal radius
$r^\star$ by a binary search takes $O(n^4 \log^2 n)$ expected
time.

\subsection{Decision algorithm}

Let $\ell(\delta, h)$ denote the line with slope $\delta$ and
$y$-intercept $h$, and let $r>0$ be fixed. We need to decide if
there exists a line $\ell(\delta, h)$ such that $\ell(\delta, h)$
intersects $D(s_i,r)$ for all $1 \le i \le n$ and the minimum
piercing number $k(\delta, h)$ for the interval system
$\{D(s_i,r)\cap \ell(\delta, h)\mid 1 \le i \le n\}$ is no more
than $k$. Note that the set $\{ D(s_i,r) : s_i \in S\}$ is fixed
and thus that if such a line $\ell(\delta, h)$ exists, then we can
move, i.e., translate and rotate, the line without making any combinatorial change, until the line reaches one of the following configurations: (i) it contacts two disks tangentially, or (ii) it passes through an intersection
of circles bounding two disks and is tangent to a third disk, or
(iii) it passes through an intersection of two disks and another
intersection of two different disks; refer to the dash-lined disks
of Figure~\ref{fig:optimal}~(c),(d) and (e), respectively. Thus we
need to check only the lines in such configurations for this
decision problem.

For the first case (i), there are only $O(n^2)$ possible configurations--the bitangents of the disks $D(s_i,r)$--so we can check them all with the decision algorithm for the fixed center line. This takes $O(n^3\log n)$ total time.  For the second case (ii), we have $O(n^3)$ possible configurations, and again we
can check them all; this takes $O(n^4\log n)$ total time. We now
describe how to check the $O(n^4)$ possible configurations
corresponding to the third case  (iii).
We fix an intersection point $o$ of circles bounding two disks,
and check all center lines passing through $o$. To this end we
first sort all the intersections of the other circles in angular
order around $o$; let $p_1, \ldots, p_m$ denote this sorted
sequence (note that $m = O(n^2)$). We denote the line passing
through $o$ and $p_i$ by $\ell(p_i)$. In a next step we determine
a maximal interval $1 \le a \le b \le m$ such that for any $a \leq
i \leq b$ the line $\ell(p_i)$ intersects all $n$ disks. This can
be done in $O(n^2 \log n)$ time by a simple angular sweep. We now
maintain the minimum piercing number for the intervals on the line
$\ell(p_i)$ as we sweep it from $\ell(p_a)$ to $\ell(p_b)$. We can
maintain the corresponding data structures in $O(\log n)$
amortized time per event as described in
Section~\ref{sec:orienfix}, so we can check in $O(n^2\log n)$ time
all center lines passing through a fixed intersection $o$. We run
this algorithm for all intersection points of circles bounding two
disks, thus we can check the $O(n^4)$ possible configurations
corresponding to the third case in $O(n^4\log n)$ time.

\subsection{Finding the optimal line and radius}

To find an optimal triple $(\delta^\star, h^\star, r^\star)$, we
use another form of binary searching which was applied to the
slope selection problem by Shafer and Steiger~\cite{shafer}. As
before, we investigate all optimal configurations to get a
discrete candidate set of optimal radius $r^\star$. Recall the
optimal configurations for the center lines of fixed orientation,
shown in Figure~\ref{fig:optimal}.
In each configuration, we can slightly rotate the center line in clockwise direction without increasing the radius while the union keeps covering $S$. Thus we need more points on another disk to immobilize the line, and we obtain optimal configurations defined by two or three disks of radius $r^\star$; in every optimal configuration with two disks, at least one of the two disks is such as (a) or (b) of Figure~\ref{fig:optimal} that contains a diametral pair of points or a triple of points on its boundary. Every other configuration consists of three disks, each having one or two points on the boundary.

From these optimal configurations, we collect candidates for the optimal radius $r^\star$. We first consider the  configurations defined by two disks. As noted above, every optimal configuration with two disks includes a diametral pair or a triple of points of $S$ that lie on the boundary of an optimal disk. This pair or triple determines the radius of the optimal disk that it lies on. So we simply compute the radii from all pairs and triples of points and get $O(n^3)$ candidates for the optimal radius. Every other configuration consists of three disks, each having one or two points on the boundary. We can interpret the radii in such configurations as intersections of triples of surfaces as follows: Let $f_{ij}(\delta, h)$ be the radius of the disk whose center lies on the line $\ell(\delta, h)$ of slope $\delta$ and $y$-intercept $h$ and whose boundary contains two points $s_i$ and $s_j$ of $S$; if $i$ is equal to $j$ then $f_{ii}(\delta, h)$ is the distance from $s_i$ to $\ell(\delta, h)$. The graph of $f_{ij}$ is a well-behaved low-degree algebraic surface in 3-dimensional $(\delta, h, r)$-space, and each triple of these surfaces provides only constant number of radius values. Let $F$ be the set of these $O(n^2)$ surfaces, and $\mathcal{A}(F)$ be their arrangement; the complexity of $\mathcal{A}(F)$ is $O(n^6)$ (see page 533 in~\cite{handbook}).

Now we may compute all $O(n^6)$ radii from the vertices of the arrangement $\mathcal{A}(F)$ and perform a binary search on the union of two sets of the previously computed $O(n^3)$ radii and of the just computed $O(n^6)$ radii, which will take $O(n^6\log n)$ time. However, we can do better if we adopt the randomization technique, as used in the slope selection problem by Shafer and
Steiger~\cite{shafer}. Instead of computing all $O(n^6)$ vertices and radii from the arrangement, we select uniformly at random $n^4$ triples of surfaces in $F$ and compute the radii from the triples; each triple of the surfaces gives us constant number of vertices from $\mathcal{A}(F)$, so we get $O(n^4)$ radii in total. We sort these $O(n^4)$ radii together with the previously computed $O(n^3)$ radii, in $O(n^4\log n)$ time. Using the decision algorithm for a fixed radius (of the previous subsection), we now perform a binary search and determine two consecutive radii $r_i$ and $r_{i+1}$ such that $r^\star$ is between $r_i$ and $r_{i+1}$. This takes $O(n^4\log^2 n)$ time in total.

Since the vertices were picked randomly, the strip $W[r_i, r_{i+1}]$ bounded by the two planes $z:= r_i$ and $z:= r_{i+1}$ contains only $O(n^3)$ vertices of $\mathcal{A}(F)$ with high probability; in fact this is always guaranteed if we select $\Omega(n^3\log n)$ triples of the surfaces. So we can compute all the vertices in $W[r_i,r_{i+1}]$ by a sweep-plane algorithm~\cite{sh02} in $O(n^4\log n)$ time as follows: we first compute the intersection of the sweeping plane at $z:= r_i$ with the surfaces in $F$. This intersection forms a two-dimensional arrangement of $O(n^2)$ quadratic closed curves and straight lines with $O(n^4)$ total complexity, so we can compute it in $O(n^4\log n)$ time. We next construct the arrangement in $W[r_i,r_{i+1}]$ incrementally by sweeping from the intersection at $z:=r_i$ towards $z:= r_{i+1}$. As a result, we can compute the $O(n^3)$ vertices (and the corresponding $O(n^3)$ radii) in $W[r_i,r_{i+1}]$. The computation time depends on the complexities of the curve arrangements on two planes $z:= r_i$ and $z:=r_{i+1}$ plus the complexity of the surface arrangement in the strip $W[r_i, r_{i+1}]$, which is $O(n^4)$. Thus the time to identify all $O(n^3)$ vertices lying in the strip is $O(n^4\log n)$.

As a final step, we perform again a binary search over these
$O(n^3)$ radii in $W[r_i,r_{i+1}]$ to find $r^\star$, which takes
$O(n^4\log^2 n)$ time. We can find the optimal radius $r^\star$ in $O(n^4\log^2 n)$ time with high probability, so this randomized algorithm to find the optimal radius takes $O(n^4\log^2 n)$ expected time. This result is summarized in the following:
\begin{theorem}\label{th:arbi}
  Let $S$ be a set of $n$ points in $\R^2$ and  $k \in \N$.  Let $r^\star$ be the minimum radius such that   there exists a set of $k$ disks of that radius   with centers on a line and with union covering $S$.  Then $r^\star$ and such disks can be computed in $O(n^4\log^2 n)$ expected time.
\end{theorem}

\section{Approximation algorithms}

We propose two approximation algorithms for the problem of computing
$k$-line centers for lines with fixed and arbitrary orientations,
respectively.

\subsection{Fixed orientation}\label{sec:fixed_apx}

We consider the case where the orientation of the line $\ell$ is given in
advance.  Without loss of generality, we may assume that $\ell$ is
horizontal. Fix an approximation parameter $\eps > 0$. Let $h$ be the
difference of the $y$-coordinates of the lowest point $L$ and the highest
point $H$ in $S$. Clearly, the optimal radius $r^\star$ is at least
$h/2$. For $\delta := \eps h$ we sample $h/\delta = 1/\eps$ lines of equal
distance between $L$ and $H$, solve the problem for the fixed lines in
$O(n\log^2 n)$ time per line, and take the smallest radius $r'$ among the
solutions. Since the optimal line lies between two consecutive sampled
lines, the radius $r'$ is at most $r^\star + \delta$, which is $r^\star+
\eps h \leq (1+2\eps)r^\star$. This result is summarized in the following:

\begin{theorem}\label{th:fptas1dof}
  Let $S$ be a set of $n$ points in $\R^2$. Let $k \in \N$ and $\eps > 0$.
  Let $r^\star$ denote the minimum radius such that there exists a set of $k$
  disks with radius $r^\star$ centered on some horizontal line that cover
  $S$. We can compute in time $O({1\over\eps}n\log^2 n)$ a set of $k$ disks
  with radius at most $(1+\eps)r^\star$ centered on some horizontal line that
  cover $S$.
\end{theorem}

\subsection{Arbitrary orientation}

In this section, we give approximation algorithms for the general case
where the line $\ell$ containing the $k$-centers is arbitrary. We first
give a constant-factor approximation algorithm, and we show how to
use this result to get a fully polynomial-time approximation scheme.

We denote by $\ell^\star$ an optimal line, and we denote by $r^\star$ the
optimal radius of $k$ disks centered at $\ell^\star$ and containing $S$.

\begin{lemma}\label{lem:constant_factor}
We can compute in time $O(n\log n)$ a radius $r_c$ and $k$ disks with radius
$r_c$ and with collinear centers, such that these disks cover $S$ and $r^\star
\leq r_c \leq \sqrt{2}r^\star$.
\end{lemma}
\begin{proof}
The width  of $S$ is the minimum distance between two lines that contain $S$;
we denote this width by $w$. We first compute a line $\ell_w$ such that the
maximum distance between $\ell_w$ and any point of $S$ is at most $w/2$; this
computation can be done in time $O(n \log n)$ by first computing the convex
hull of $S$, and then by finding the width of this convex hull in linear time
using the rotating calliper technique~\cite{toussaint}.

\epsfigure[\textwidth]{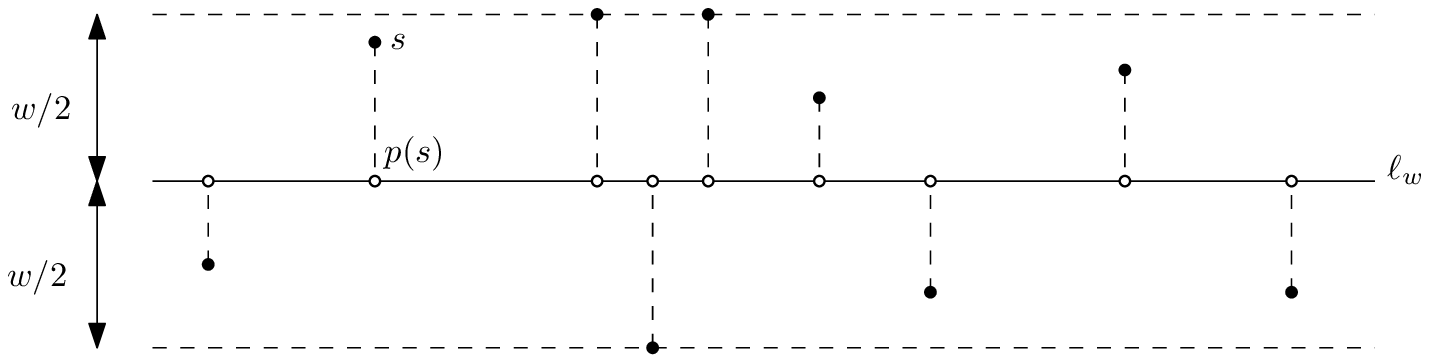}
{Proof of Lemma~\ref{lem:constant_factor}. The black dots are the points in $S$, and the white dots are the points in $S_w=p(S)$.}{fig:constant}

We denote by $p$ the orthogonal projection to $\ell_w$. We project all  the
points in $S$ and obtain a point set $S_w =\{ p(s) \mid s\in S\}$. (See
Figure~\ref{fig:constant}.) We solve the $k$-center problem for $S_w$ when the
centers are constrained to lie on $\ell_w$. It is the one-dimensional
$k$-center problem, which can be solved in $O(n)$ time after the points in
$S_w$ have been sorted~\cite{frederickson}. We denote by $r_w$ the optimal
radius for this problem, and we denote by $C_w$ a set of $k$ points such that
$S$ is contained in the union of the $k$ disks with radius $r_w$ and center in
$C_w$. We have $r_w \leq r^\star$, because when we project the $k$ disks of a
solution of the original problem to $\ell_w$, we obtain a set of segments with
length $2r^\star$ whose union contains $S_w$. We now distinguish between two
cases.

First we assume that $r_w \leq w/2$. Let $s$ be a point in $S$. There exists $c \in C_w$ such that $|cp(s)| \leq r_w$. As $|sp(s)| \leq w/2$, it implies that $|cs| \leq w/\sqrt{2}$. We have just proved that $S$ is contained in the union of the disks centered at $C_w$ with radius $w/\sqrt{2}$. As we noticed earlier
that $r^\star \geq w/2$, we conclude that the set of disks centered at $C_w$
with radius $w/\sqrt{2}$ is a $\sqrt{2}$-factor approximation of  the optimum.

Now we prove the remaining case: we assume that $r_w > w/2$. Let $s$ be a point
in $S$. There exists $c \in C_w$ such that $cp(s) \leq r_w$. Since $|sp(s)|
\leq w/2$, we get $|cs| \leq \sqrt{2}r_w$. It follows that $S$ is contained in
the union of the disks centered at $C_w$ with radius $\sqrt{2}r_w$, and we
conclude using the fact that $r_w \leq r^\star$.
\end{proof}

We now extend Lemma~\ref{lem:constant_factor} into an approximation scheme. We
first compute a radius $r_c$ such that $r^\star \leq r_c \leq \sqrt{2}r^\star$.
The diameter $d$ of $S$ is the maximum distance between any two points in $S$.
We compute a pair $a,b \in S$ such that $|ab|=d$; it can be done in $O(n \log
n)$ time in the same way as we computed the width. The lemma below handles the
case where $S$ is skinny.
\begin{lemma}\label{lem:skinny}
We assume that $0 <\eps <1$. If $d \geq 3r_c$, then we can compute in time
$O(\eps^{-2}n\log^2 n)$ a set of $k$ disks with collinear centers and with
radius less than $(1+\eps)r^\star$, such that these disks cover $S$.
\end{lemma}
\begin{proof}
Let $\lambda>0$ be a constant, to be specified later. We scale $S$ so that
$d=1$, and we choose a coordinate frame such that $a=(0,0)$ and $b=(1,0)$. When
$i \in \Z$, we denote $m_i=i\lambda\eps r_c$. For each $i$ such that $-6r_c
\leq m_i \leq 6r_c$, using the result of Section~\ref{sec:fixed_apx}, we
compute a $(1+\eps/3)$-approximation for lines with slope $m_i$, and we return
the covering with disks of smallest radius. As we consider only $O(1/\eps)$
slopes $m_i$, it takes total time $O(\eps^{-2}n\log^2 n)$. We will now prove
the correctness of this algorithm.

\epsfigure[\textwidth]{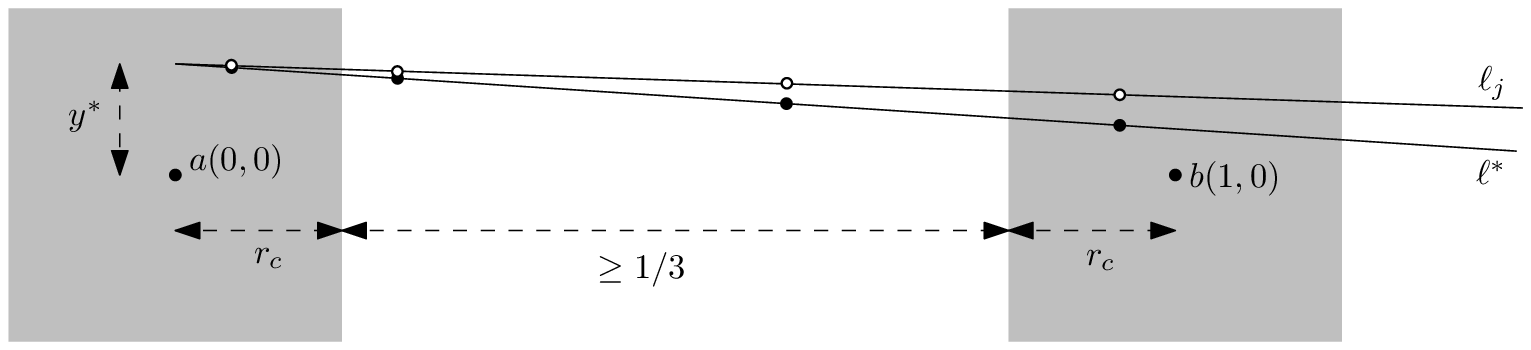} {The black dots along $\ell^\star$ represent
$C^\star$, and the white dots represent $C$.} {fig:skinny}

Let $\ell^\star$ be an optimal line; we write its equation $y=m^\star
x+y^\star$. Consider the two axis-parallel squares centered at $a$ and $b$ with edge-length $2r_c$. (See Figure~\ref{fig:skinny}.) Since $|ab|=d=1$ and $d \geq 3r_c$, a line with slope  outside the interval $[-6r_c,6r_c]$ cannot intersect both these squares, and thus this line is at distance more than $r_c$ from $a$ or $b$. Since $r^\star \leq r_c$, we know that $\ell^\star$ intersects both of these squares and thus $m^\star \in [-6r_c,6r_c]$. So there exists $j \in \Z$ such that $m_j \in [-6r_c,6r_c]$ and $|m_j-m^\star|\leq \lambda \eps r_c$.

We denote by $C^\star \subset \ell^\star$ the centers of $k$ disks with radius
$r^\star$ that cover $S$. All the points in $S$ have $x$-coordinates in
$[0,1]$, so we can assume that the points in $C^\star$ also have
$x$-coordinates in $[0,1]$. Let $x_1,\dots,x_k \in [0,1]$ denote the
$x$-coordinates of the points in $C^\star$. We have
\[C^\star=\{c^\star_i\mid c^\star_i=(x_i,m^\star x_i+y^\star) \mbox{ and } 1\leq i \leq k.\}\]
We introduce the point set $C$ obtained by translating each point of $C^\star$
vertically until it reaches the line $\ell_j$ with equation $y=m_jx+y^\star$.
Hence we have
\[C=\{c_i\mid c_i=(x_i,m_jx_i+y^\star) \mbox{ and } 1\leq i \leq k.\}\]
Notice that for all $i$, we have
\[|c^\star_ic_i| \leq |m^\star-m_j| \leq \lambda \eps r_c \leq \lambda \eps \sqrt{2} r^\star.\]
So choosing $\lambda=1/(3\sqrt{2})$, we get that $|c^\star_ic_i|\leq \eps r^\star/3$ and thus $S$ is covered by the disks with radius $(1+\eps/3)r^\star$
centered at $C \subset \ell_j$. So if we consider the $(1+\eps/3)$-approximate
radius $r_j$ that our algorithm returned for lines with slope $m_j$, we have
\[r_j \leq (1+\eps/3)(1+\eps/3)r^\star < (1+\eps) r^\star,\]
which completes our proof.
\end{proof}

In the following lemma, we handle the remaining case where $S$ is fat.

\epsfigure[10cm]{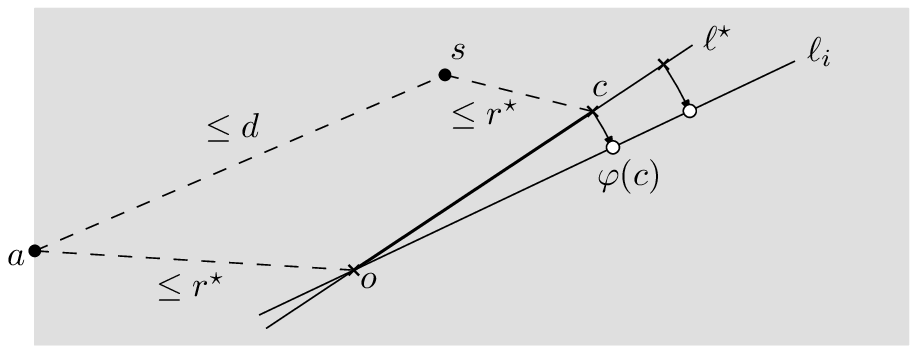} {Proof of Lemma~\ref{lem:fat}. The crosses along
$\ell^\star$ represent $C^\star$, and the white dots represent $C$.}{fig:fat}

\begin{lemma}\label{lem:fat}
We assume that $0 <\eps <1$. If $d \leq 3r_c$, then we can compute in time
$O(\eps^{-2}n\log^2 n)$ a set of $k$ disks with collinear centers and with
radius less than $(1+\eps)r^\star$, such that these disks cover $S$.
\end{lemma}
\begin{proof}
Let $\lambda>0$ be a constant, to be specified later. For each integer $i \geq
0$ such that $\theta_i=i\lambda\eps \leq \pi$, using the result of
Section~\ref{sec:fixed_apx}, we compute a $(1+\eps/3)$-approximation for lines
making an angle $\theta_i$ with horizontal. Among all these
$(1+\eps/3)$-approximate solution, we return one with minimum radius. As there
are $O(1/\eps)$ angles $\theta_i$, it takes total time $O(\eps^{-2}n\log^2 n)$.
We will now prove the correctness of this algorithm.

We denote by $C^\star$ a set of centers in an exact solution to our problem,
and we denote by $\theta^\star \in [0,\pi)$ the angle that the corresponding
optimal line $\ell^\star$ makes with horizontal. Since $r^\star \leq d$, there
exists a point $o \in C^\star$ that is at distance at most $d$ from $a$. We
choose $i$ such that $|\theta_i-\theta^\star|$ is minimized, hence
$|\theta_i-\theta^\star| \leq \lambda\eps$. We consider the line $\ell_i$
through $o$ making the angle $\theta_i$ with horizontal. We denote by $\varphi$
the rotation around $o$ with angle $\theta_i-\theta^\star$; hence
$\varphi(\ell^\star)=\ell_i$.

We now prove that the disks with radius $(1+\eps/3)r^\star$ centered at
$\varphi(C^\star)$ cover $S$. So let $s$ denote a point in $S$. There exists a
center $c \in C^\star$ such that $|cs| \leq r^\star$. We have $|oc| \leq
|oa|+|as|+|sc|$ and thus $|oc| \leq 3d$. It follows that $|c\varphi(c)| \leq 3d
|\theta_i-\theta^\star|$ and thus $|c\varphi(c)| \leq 3d\lambda\eps$. Then we
have $|\varphi(c)s|\leq r^\star+3d\lambda\eps$, and since $d \leq 3r_c$ and
$r_c \leq \sqrt{2}r^\star$, we get $|\varphi(c)s| \leq
(1+9\sqrt{2}\lambda\eps)r^\star$. Choosing $\lambda=1/(27\sqrt{2})$, we get
that the disks with radius $(1+\eps/3)r^\star$ centered at $\varphi(C^\star)$
cover $S$.

Consider the approximate $k$-center $C_i$ that was computed for lines making an
angle $\theta_i$ with horizontal, and $r_i$ the corresponding radius. We have
$r_i \leq (1+\eps/3)(1+\eps/3)r^\star < (1+\eps)r^\star$, which proves that
$C_i$ is a $(1+\eps)$-approximate $k$ center for lines with arbitrary
direction.
\end{proof}

Combining lemmas~\ref{lem:skinny} and~\ref{lem:fat}, we obtain the following
fully polynomial-time approximation scheme:
\begin{theorem}\label{th:fptas}
Let $S$ be a set of $n$ points in $\R^2$. Let $k \in \N$ and $0 <\eps <1$. Let
$r^\star$ denote the minimum radius such that there exists a set of $k$ disks
with radius $r^\star$ and collinear centers that cover $S$. We can compute in
time $O(\eps^{-2}n\log^2 n)$ a set of $k$ disks with radius at most
$(1+\eps)r^\star$ and collinear centers that cover $S$.
\end{theorem}


\begin{thebibliography}{0}


\bibitem{agarwal02}
  P.~K.~Agarwal and C.~M.~Procopiuc.
   Exact and Approximation Algorithms for Clustering.
{\em Algorithmica} 33(2), 201--226, 2002.

\bibitem{alt}
    H.~Alt, E.~M.~Arkin, H.~Br{\"o}nnimann, J.~Erickson,
    S.~P.~Fekete, C.~Knauer, J.~Lenchner, J.~S.~B.~Mitchell, and K.~Whittlesey. Minimum-cost Coverage of point sets by disks.
In {\em Proc. of the 22nd Annual ACM Symposium on Computational Geometry} (2006), 449--458.

\bibitem{bose08}
    P. Bose, S. Langerman and S. Roy.
    Smallest enclosing circle centered on a query line segment.
    In \emph{Proc. of CCCG} (2008), 13--15.

\bibitem{bose}
    P. Bose and G. Toussaint.
    Computing the constrained Euclidean, geodesic and link center of a simple polygon with applications.
    In {\em Proc. of the Pacific Graphics International} (1996), 102--112.


\bibitem{chan}
    T.~M.~Chan and A.-Al Mahmood.
    Approximating the piercing number for unit-height rectangles.
In {\em Proc. of Canadian Conference on Computational
Geometry} (2005), 15--18.

\bibitem{Chan99}
  T.~Chan.
  Geometric applications of a randomized optimization technique.
  {\em Discrete {\&} Computational Geometry} 22(4), 547--567, 1999.

\bibitem{cole}
    R.~Cole, J.~Salowe, W.~Steiger, and E.~Szemer{\'e}di,
    An optimal-time algorithm for slope selection.
    {\it SIAM Journal on Computing} 18, 792--810, 1989.

\bibitem{das}
    G. K. Das, S. Roy, S. Das, and S. C. Nandy.
    Variations of Base Station Placement Problem on the Boundary of a Convex Region.
    In \emph{Workshop on Algorithms and Computation}, (2007) 151-152.

\bibitem{D95}
  Z.~Drezner.
  {\em Facility Location}. Springer-Verlag, 1995.

\bibitem{Epp97}
D.~Eppstein. Fast construction of planar two-centers. In {\em
Proc. of the 8th Annual ACM-SIAM Symposium on Discrete
Algorithms} (1997), 131--138.

\bibitem{FPT81}
R.~Fowler, M.~Paterson, and S.~Tanimoto.
 Optimal packing and covering in the plane are {NP}-complete.
 {\em Information Processing Letters} 12(3), 133--137, 1981.

\bibitem{frederickson}
    G.~N.~Frederickson.
    Optimal algorithms for tree partitioning.
In {\em Proc.  of the Annual ACM-SIAM Symposium on Discrete
Algorithms} (1991), 168--177.

\bibitem{fredjohn}
    G.~N.~Frederickson and D.~B.~Johnson.
    Generalized selection and ranking: sorted matrices.
    {\em SIAM Journal on Computing} 13, 14--30, 1984.

\bibitem{handbook}
    J.~E.~Goodman and J.~O'Rourke.
    \emph{Handbook of Discrete and Computational Geometry}.
    Second Edition, CRC Press, 2004.

\bibitem{hurtado}
    F. Hurtado and G. Toussaint.
    Facility location problems with constraints
    \emph{Studies of Locaation Analysis, Special Issue on Comp. Geom.}
    (2000), 15--17.

\bibitem{hwang}
    R.~Z.~Hwang,  R.~C.~T.~Lee, and R.~C.~Chang.
    The slab dividing approach to solve the Euclidean $p-$center problem.
    {\em Algorithmica} 9, 1--22, 1993.

\bibitem{katz}
    M.~J.~Katz, F.~Nielsen, and M.~Segal.
    Maintenance of a piercing set for intervals with applications.
    {\em Algorithmica} 36, 59--73, 2003.

\bibitem{matousek}
    J.~Matou\v{s}ek.
    Randomized optimal algorithm for slope selection.
    {\em Information Processing Letters} 39, 183--187, 1991.

\bibitem{Meg86}
N.~Megiddo.
 Linear-time algorithms for linear programming in $R^3$ and related problems.
 {\em SIAM Journal on Computing} 12, (1983) 759--776.

\bibitem{Meg90}
N.~Megiddo.
 On the complexity of some geometric problems in unbounded dimension.
 {\em Journal of Symbolic Computation} 10, 327--334, 1990.

\bibitem{MS84}
N.~Megiddo and K.~Supowit.
 On the complexity of some common geometric location problems.
 {\em SIAM Journal on Computing} 13, 182--196, 1984.

\bibitem{roy}
    S. Roy, d. Bardhan and S. Das.
    Efficient algorithm for placing base stations by avoiding forbidden zone.
    In \emph{Proc. of the 2nd Int. Conf. on Distributed Computing and Internet Technology} Springer LNCS, 3816 (2005) 105-116.

\bibitem{shafer}
    L.~Shafer and W.~Steiger.
    Randomized optimal geometric algorithms.
In {\em Proc. of  Canadian Conference on Computational
Geometry} (1993), 133--138.

\bibitem{sh02}
    H.~Shaul and D.~Halperin.
    Improved construction of vertical decompositions of
    three-dimensional arrangements.
In  {\em Proc. of the the 22nd Annual ACM Symposium on
Computational Geometry} (2002), 283--292.

\bibitem{chansu}
    C.-S. Shin, J.-H. Kim, S. K. Kim and K.-Y. Chwa.
    Two-Center Problems for a Convex Polygon.
    In \emph{Proc. of the 6th Annual European Symposium on Algorithms} (1998), 199--210.

\bibitem{dtree}
  D.~Sleator and R.~Tarjan.
  A data structure for dynamic trees.
  {\em Journal of Computer and System Sciences} 26(3), 362--391, 1983.

\bibitem{toussaint}
  G.~T.~Toussaint.
  Solving geometric problems with the rotating calipers.
In {\em Proc. of IEEE MELECON} (1983), 1--8.

\end{thebibliography}
\end{document}